\documentclass[a4paper,onecolumn,11pt,accepted=2026-08-03]{quantumarticle}
\pdfoutput=1
\usepackage[utf8]{inputenc}
\usepackage[english]{babel}
\usepackage[T1]{fontenc}

\usepackage{algorithm}
\usepackage{algpseudocode}
\usepackage{amsmath,amssymb,amsthm,mathtools}
\usepackage[numbers]{natbib}
\usepackage[table]{xcolor}
\usepackage{graphicx}
\usepackage{tikz}
\usetikzlibrary{shadings}
\usepackage{pgfplots}
\pgfplotsset{compat=1.18}
\usepackage{hyperref}
\usepackage{enumitem}
\usepackage{url}

\mathtoolsset{showonlyrefs}

\newcommand{\Ftwo}{\mathbb{F}_2}
\newcommand{\Z}{\mathbb{Z}}
\newcommand{\ZP}{\mathbb{Z}_P}
\newcommand{\ZPx}{\mathbb{Z}_P^{\times}}
\newcommand{\Ker}{\mathrm{Ker}}

\newcommand{\idxLtwo}{[L/2]}

\newcommand{\Creq}{\cellcolor{gray!20}1}

\newcommand{\Hful}{\hat{H}}
\newcommand{\Hres}{\tilde{H}}
\DeclareMathOperator{\diag}{diag}
\DeclareMathOperator{\Row}{Row}

\DeclareMathOperator{\wt}{wt}

\makeatletter
\renewcommand\[{\begin{equation}}
\renewcommand\]{\end{equation}}
\makeatother

\theoremstyle{definition}

\newtheorem{theorem}{Theorem}[section]

\title{Breaking the Orthogonality Barrier in Quantum LDPC Codes}
\author{Kenta Kasai}
\email{kenta@ict.eng.isct.ac.jp}
\affiliation{Institute of Science Tokyo}
\begin{document}
\maketitle

\begin{abstract}
Classical low-density parity-check (LDPC) codes are a widely deployed and well-established technology, forming the backbone of modern communication and storage systems.
It is well known that, in this classical setting, increasing the girth of the Tanner graph while maintaining regular degree distributions leads simultaneously to good belief-propagation (BP) decoding performance and large minimum distance.
In the quantum setting, however, this principle does not directly apply because quantum LDPC codes must satisfy additional orthogonality constraints between their parity-check matrices.
When one enforces both orthogonality and regularity in a straightforward manner, the girth is typically reduced and the minimum distance becomes structurally upper bounded.
In this work, we overcome this limitation by using permutation matrices with controlled commutativity and by restricting the orthogonality constraints to only the active part of the construction, while preserving regular check-matrix structures. This design circumvents conventional structural distance limitations induced by parent-matrix orthogonality, and enables the construction of quantum LDPC codes with large girth while avoiding latent low-weight logical operators.
As a concrete demonstration, we construct a girth-8, (3,12)-regular \([[9216,4612, \leq 48]]\) quantum LDPC code and show that, under BP decoding combined with a low-complexity post-processing algorithm, it achieves a frame error rate as low as  \(10^{-8}\) on the depolarizing channel with error probability \(4 \%\).
\end{abstract}

\section{Introduction}
Classical low-density parity-check (LDPC) codes are a widely deployed and well-established technology in modern communication and storage systems~\cite{Gallager1963}, and   belief-propagation (BP)  decoding can achieve channel capacity~\cite{KudekarRichardsonUrbanke2013,RichardsonUrbanke2001}. Performance depends strongly on the Tanner-graph degree distribution: regular codes can perform well, and carefully designed irregular distributions can further improve BP performance~\cite{RichardsonShokrollahiUrbanke2001}, whereas naive irregularization can degrade it. Random LDPC codes with variable-node degree at least 3 have minimum distance growing linearly with blocklength~\cite{RichardsonUrbanke2008,Gallager1963}, so classical LDPC research has not focused on increasing minimum distance itself. Moreover, BP decoding can be trapped by small Tanner-graph structures called trapping sets~\cite{Richardson2003ErrorFloors}, which cause decoding failures; suppressing harmful trapping sets typically requires large girth~\cite{Tanner1981}.

Sparse-graph quantum error-correcting codes were proposed in~\cite{MacKay2004}, and the behavior of iterative decoding and the role of degeneracy were discussed in~\cite{PoulinChung2008}. Families with positive rate and distance $\Theta(\sqrt{n})$ are known~\cite{TillichZemor2014}. Constructions based on Kronecker sum and product~\cite{KovalevPryadko2013}, distance bounds for generalized bicycle  codes~\cite{WangPryadko2022}, finite-length evaluations~\cite{Mostad2025}, and decoding protocols~\cite{Lin2025} have been reported, advancing the theoretical understanding of   Calderbank--Shor--Steane (CSS)-type quantum LDPC  codes. However, the orthogonality constraint between the CSS parity-check matrices $H_X$ and $H_Z$, namely $H_X H_Z^{\mathsf T}=0$, has no classical counterpart, and enforcing both orthogonality and regularity in a straightforward manner typically reduces girth and induces structural distance upper bounds.

Surface codes and hypergraph product  codes rely on geometric or product structures and thus follow design principles different from those developed for classical LDPC codes. As a result, it is not straightforward to apply classical degree-distribution and girth design techniques directly~\cite{Fowler2012,TillichZemor2014,RichardsonUrbanke2008}. The goal of this work is to provide a construction principle that preserves classical LDPC parity-check structures while enabling their direct use as CSS check matrices. Specifically, we aim for regular LDPC codes with large girth and variable degree at least 3, whose minimum distance is not trivially upper bounded.

Prior work has explored cyclic parity-check matrices that allow explicit control of regular Tanner-graph degree distributions. Representative constructions include the Hagiwara--Imai codes~\cite{HagiwaraImai2007} and bicycle constructions~\cite{MacKay2004,KovalevPryadko2013,WangPryadko2022}. However, these constructions have known limitations. Even for classical codes, one-step lifting of a protograph  by circulant permutation matrices (CPMs)  , i.e., one-step CPM lifting,  can impose fixed upper bounds on girth due to the protograph structure. Mitchell et al.\ pointed out that one-step lifting can impose fixed bounds on minimum distance and girth and that two-step lifting can improve them~\cite{MitchellSmarandacheCostello2014}. For CSS codes   built by CPM lifts, the  Tanner-graph girth satisfies $g=\min\{g(H_X),g(H_Z)\}$, so   CPM-lift  girth bounds apply to each matrix. For   CPM-lift  CSS designs, generalized CSS (entanglement-assisted) constructions use classical   CPM-lift  constructions with girth $\ge 6$ to avoid 4-cycles~\cite{HsiehBrunDevetak2009}. Explicit constructions achieving girth 12 for column weight 2 have also been reported~\cite{KomotoKasai2025}. Moreover, for general quantum CPM--LDPC codes, the girth is upper bounded by 12 for column weight 2 and by 6 for column weight at least 3~\cite{AmirzadePanarioSadeghi2024}. Using   affine-permutation-matrix (APM)--LDPC  codes, these upper bounds can be broken~\cite{KomotoKasai2025,KasaiISTC2025}.

In CSS-type LDPC constructions, deleting rows from the parity-check matrices $H_X$ or $H_Z$ preserves the commutation condition $H_X H_Z^{\mathsf T}=0$, but it weakens stabilizer constraints and can enlarge the code space~\cite{MacKay2004,HagiwaraImai2007}. Row deletion is used for rate adjustment and for controlling row and column weights, yet distance preservation is not guaranteed in general~\cite{HagiwaraImai2007,Ostrev2024}. In particular, deleting rows can allow low-weight vectors to enter $C_X\setminus C_Z^{\perp}$ or $C_Z\setminus C_X^{\perp}$, creating new low-weight logical operators and thus reducing the minimum distance. In CSS codes, with $C_X=\ker(H_X)$ and $C_Z=\ker(H_Z)$, low-weight vectors in $C_X\setminus C_Z^{\perp}$ or $C_Z\setminus C_X^{\perp}$ become non-trivial logical operators~\cite{CalderbankShor1996,Steane1996,Gottesman1997}. Consequently, deleted check rows (which are LDPC and thus low weight) can become logical operators and upper bound the distance by the row weight. In the CSS product-code constructions of Ostrev et al., row removal is explicitly used to control the number of stabilizers, and distance degradation can occur~\cite{Ostrev2024}. Lin et al.\ report that removing redundant stabilizer rows reduces the syndrome distance, which is a structural effect due to fewer constraints, not a decoder artifact~\cite{Lin2025}.

In this work, we first discuss the general mechanism by which row deletion can degrade the minimum distance of CSS codes. Next, for generalized Hagiwara--Imai codes~\cite{KomotoKasai2025}, we establish a code-design method that controls commutativity of submatrices to avoid distance degradation due to row deletion. More concretely, using the algebraic structure of APM--LDPC codes, we guarantee orthogonality only on the active part, introduce a set of row pairs for which commutation is required, and localize the commutation condition. We further control the nonzero pattern of the interaction matrix so that deleted rows do not materialize as low-weight logical operators, and present a sequential construction algorithm based on the APM composition rule. We also give a girth upper bound under the required commutativity and explicitly construct a code that attains this bound.

As a concrete demonstration, we construct a (3,12)-regular \([[9216,4612,\leq 48]]\) quantum LDPC code and show that, under BP decoding combined with a low-complexity post-processing algorithm, it achieves a frame error rate (FER) as low as \(10^{-8}\) on the depolarizing channel with error probability \(4\%\). We also compare constructions with and without commutation control to demonstrate suppression of low-weight logical operators.

\section{Problem setting: low-weight logical operators induced by check-row removal}
\label{sec:deleted-checks-logicals}

We abstract the conventional CSS--LDPC construction used in this paper, based on~\cite{MacKay2004,HagiwaraImai2007}. We first construct two orthogonal square (block-circulant) parent matrices $\Hful_X,\Hful_Z$. The CSS code defined by these parent matrices typically has rate zero, so to adjust the rate we delete internal rows and use the remaining rows as the active matrices $H_X$ and $H_Z$. Namely,
\[
\Hful_X=
\begin{bmatrix}
H_X\\
\Hres_X
\end{bmatrix},\qquad
\Hful_Z=
\begin{bmatrix}
H_Z\\
\Hres_Z
\end{bmatrix}
\]
We refer to $H_X$ and $H_Z$ as the active part and $\Hres_X$ and $\Hres_Z$ as the latent part.

However, in this construction, parent orthogonality imposes strong constraints on the latent rows and can degrade the minimum distance. The minimum distance is defined as $d_{\min}=\min\{d_X,d_Z\}$ with
\[
 d_Z:=\min\{\wt(\boldsymbol{z}):\boldsymbol{z}\in C_X\setminus C_Z^{\perp}\},\qquad
 d_X:=\min\{\wt(\boldsymbol{x}):\boldsymbol{x}\in C_Z\setminus C_X^{\perp}\}.
\]
Parent orthogonality
\[
\Hful_X(\Hful_Z)^{\mathsf T}=0
\]
forces not only the active orthogonality between the active matrices $H_X$ and $H_Z$, namely $H_XH_Z^{\mathsf T}=0$, but also
\begin{align}
 H_X(\Hres_Z)^{\mathsf T}=0,\qquad H_Z(\Hres_X)^{\mathsf T}=0\label{021224_10Jan26}
\end{align}
which implies $\Row(\Hres_X)\subset C_Z$ and $\Row(\Hres_Z)\subset C_X$. Thus each row $\boldsymbol{x}$ of $\Hres_X$ satisfies $\boldsymbol{x}\in C_Z$, and each row $\boldsymbol{z}$ of $\Hres_Z$ satisfies $\boldsymbol{z}\in C_X$. In general $\boldsymbol{x}$ need not belong to $C_X^{\perp}$ and $\boldsymbol{z}$ need not belong to $C_Z^{\perp}$, in which case they become logical operators. The minimum row weight (in our construction, $L$) becomes an upper bound on $d_X$ and $d_Z$.

In this work we also obtain the active matrices $H_X,H_Z$ by deleting rows from the parent matrices, but we design them so that low-weight rows in the latent matrices are not orthogonal to $H_Z$ or $H_X$, respectively. That is, for low-weight $\boldsymbol{x}\in\Row(\Hres_X)$ we enforce $H_Z\boldsymbol{x}^{\mathsf T}\neq 0$, and for low-weight $\boldsymbol{z}\in\Row(\Hres_Z)$ we enforce $H_X\boldsymbol{z}^{\mathsf T}\neq 0$, so that in particular
$\Row(\Hres_X)\not\subset C_Z$ and $\Row(\Hres_Z)\not\subset C_X$.
The same argument applies not only to individual latent rows but also to low-weight linear combinations.

Accordingly, we define latent-based distance for the $X$ and $Z$ distances as
\begin{align}
 d_X^{(\mathrm{lat})} &:= \min\{\wt(\boldsymbol{x}) : \boldsymbol{x}\in (C_Z\cap\Row(\Hres_X))\setminus C_X^{\perp}\},\\
d_Z^{(\mathrm{lat})} &:= \min\{\wt(\boldsymbol{z}) : \boldsymbol{z}\in (C_X\cap\Row(\Hres_Z))\setminus C_Z^{\perp}\}. 
\end{align}
These quantities upper bound $d_{\min}$.
The proposed design aims to make these bounds as large as possible. In~\cite{Kasai2025Hashing}, binary submatrices corresponding to codewords are lifted to a nonbinary field so that they no longer represent codewords; this can be viewed as increasing these latent-based distance via nonbinary lifting.

\section{Method of matrix design}
\label{sec:general-theory}
In this section, based on the above problem setting, we aim to construct the latent matrices $\Hres_X$ and $\Hres_Z$ such that
\begin{align}
 H_XH_Z^{\mathsf T}=0, \qquad H_X(\Hres_Z)^{\mathsf T}\neq 0,\qquad H_Z(\Hres_X)^{\mathsf T}\neq 0,
\end{align}
and, more precisely, such that $d_X^{(\mathrm{lat})}$ and $d_Z^{(\mathrm{lat})}$ are large.
We define generalized Hagiwara--Imai codes, i.e., parent matrix pairs with block-circulant structure, and develop a general theory for imposing only the active orthogonality $H_XH_Z^{\mathsf T}=0$ between the active matrices $H_X$ and $H_Z$.

A protograph LDPC code with column weight $J$ and row weight $L$ is defined by a parity-check matrix consisting of $J\times L$ permutation matrices of size $P$~\cite{Thorpe2003}. For a permutation $f:[P]\to[P]$, the corresponding $P\times P$ permutation matrix $F=P(f)$ is defined by
\(
F_{x,y}=1 \text{ if and only if } f(x)=y.
\)
Each row and column has exactly one 1, so the parent matrix is sparse while its block structure is compact.

We introduce generalized Hagiwara--Imai codes~\cite{KomotoKasai2025} as protograph generalizations of the Hagiwara--Imai   CPM-lift  CSS codes~\cite{HagiwaraImai2007}. For permutation matrices $F_i,G_i\ (i\in[L/2])$ of size $P$, define the parent matrices $\Hful_X,\Hful_Z$ of size $(LP/2)\times(LP)$ by
\[
(\Hful_X)_{i,j}=F_{j-i},\quad
(\Hful_X)_{i,L/2+j}=G_{j-i},
\]
\[
(\Hful_Z)_{i,j}=(G_{i-j})^{\mathsf T},\quad
(\Hful_Z)_{i,L/2+j}=(F_{i-j})^{\mathsf T}.
\]
 Throughout, indices in $\idxLtwo$ are taken modulo $L/2$, i.e., we identify $\idxLtwo$ with $\mathbb{Z}_{L/2}$. A sufficient condition for orthogonality is given below. The active matrices $H_X$ and $H_Z$ consist of the top $J$ block rows. Conventional designs \cite{KomotoKasai2025, Kasai2025Hashing} impose commutativity of $F_i$ and $G_j$ so that both $H_X,H_Z$ and $\Hful_X,\Hful_Z$ are orthogonal.
The top $J$ block rows of the parent matrices are used as the active matrices $H_X$ and $H_Z$. Our interest is to satisfy active orthogonality $H_XH_Z^{\mathsf T}=0$ while not imposing parent-matrix orthogonality.

This focus on active orthogonality is motivated by the desire to keep the latent part non-orthogonal to the active part, so that low-weight latent combinations do not automatically become logical operators. The block-circulant structure makes $\Hful_X(\Hful_Z)^{\mathsf T}$ depend only on differences, and the resulting interaction matrices $\Psi_r$ allow us to localize commutativity constraints to a small difference set.

Each block row of $\Hful_X$ is a cyclic shift of $(F_0,F_1,\dots,F_{L/2-1})$ in the left half and $(G_0,G_1,\dots,G_{L/2-1})$ in the right half. Hence blocks depend only on the difference $(j-i)$, and the product of parent matrices depends only on the difference. For any $i,k\in[L/2]$, the $(i,k)$ block is
\[
\bigl(\Hful_X(\Hful_Z)^{\mathsf T}\bigr)_{i,k}=\sum_{u=0}^{L/2-1}\Bigl(F_uG_{r-u}+G_{r-u}F_u\Bigr)=:\Psi_r
\]
which depends only on $r:=(k-i)\bmod\, L/2$. For $r\in\idxLtwo$, if $F_u$ and $G_{r-u}$ commute for all $u\in\idxLtwo$, then $\Psi_r=0$.

\subsection{Sufficient condition for active orthogonality}
We re-derive a convenient form of a sufficient condition for $H_XH_Z^{\mathsf T}=0$~\cite{Kasai2025Hashing}. Define
\[
\Delta:=\{(k-i)\bmod\, L/2\mid 0\le i,k\le J-1\}\subseteq[L/2].
\]
The next theorem spells out how commutativity restricted to these differences guarantees active orthogonality.
\begin{theorem}
\label{thm:active-orthogonality}
If $F_u$ and $G_{r-u}$ commute for all $r\in\Delta$ and all $u\in[L/2]$, then $H_XH_Z^{\mathsf T}=0$.
\end{theorem}
\begin{proof}
For any $i,k\in[J]$, we have $r=(k-i)\bmod\, L/2\in\Delta$. By assumption, $\Psi_r=0$, hence for all $i,k\in[J]$,
\(
(\Hful_X(\Hful_Z)^{\mathsf T})_{i,k}=\Psi_{(k-i)\bmod\, L/2}=0.
\)
Therefore $H_XH_Z^{\mathsf T}=0$.
\end{proof}

\subsection{Necessary condition for latent non-orthogonality}
Define the index pairs of $(F_i,G_j)$ required to commute by
\begin{align*}
\Gamma :=\bigcup_{r\in \Delta}\Gamma_r,\qquad \Gamma_r :=\{(i,j)\mid (i,j)=(u,r-u),\ u\in[L/2]\}
\end{align*}
where subscripts are modulo $L/2$. 
Here $\Gamma_r$ is the set of index pairs contributing to $\Psi_r$, and $\Gamma_r\cap\Gamma_s=\varnothing$ for $r\neq s$.

  The intended high-rate regime has $J<L/2$.
Indeed, if the active block rows have full row rank, then $J\ge L/2$ gives zero or negative design rate.
The full-rank assumption is not automatic for every finite instance; for example, the explicit code in Section~\ref{sec:example} has active ranks slightly below the design ranks.
Thus the condition $J<L/2$ should be read as a conservative design condition for this construction, not as a necessary rank theorem for all possible degenerate choices.
 We first state a basic feasibility constraint: without enough blocks, latent non-orthogonality cannot be forced.
\begin{theorem}
Assume the active orthogonality $H_XH_Z^{\mathsf T}=0$. For the standard active choice (top $J$ block rows), to have $H_X(\Hres_Z)^{\mathsf T}\neq 0$ and $H_Z(\Hres_X)^{\mathsf T}\neq 0$, it is necessary that $L\ge 4J$.
\end{theorem}
\begin{proof}
Assume $L<4J$. Then for this active choice, $\Delta=[L/2]$. Active orthogonality $H_XH_Z^{\mathsf T}=0$ implies $\Psi_r=0$ for all $r\in\Delta$, hence $\Psi_r=0$ for all $r\in[L/2]$. Therefore every block of $\Hful_X(\Hful_Z)^{\mathsf T}$ vanishes, in particular the mixed blocks between active and latent rows, so $H_X(\Hres_Z)^{\mathsf T}=0$ and $H_Z(\Hres_X)^{\mathsf T}=0$, contradicting the premise.
\end{proof}
Intuitively, when $L<4J$ the difference set $\Delta$ covers all residues, so enforcing active orthogonality forces $\Psi_r=0$ for every $r$. This annihilates all mixed active--latent blocks and makes latent rows orthogonal, which is exactly what we aim to avoid.

We next isolate the minimal algebraic obstruction to latent non-orthogonality.
\begin{theorem}
To have $H_X(\Hres_Z)^{\mathsf T}\neq 0$ and $H_Z(\Hres_X)^{\mathsf T}\neq 0$, it is necessary that there exists $r\in [L/2]\setminus \Delta$ with $\Psi_r\neq 0$.
\end{theorem}
\begin{proof}
Assume $\Psi_r=0$ for all $r\in [L/2]\setminus \Delta$. For any $i\in[L/2]\setminus[J]$ and $k\in[J]$, $r=(k-i)\bmod\, L/2$ does not belong to $\Delta$. By assumption $\Psi_r=0$, hence for all such $i,k$,
\(
(\Hful_X(\Hful_Z)^{\mathsf T})_{i,k}=\Psi_{(k-i)\bmod\, L/2}=0.
\)
Therefore $H_X(\Hres_Z)^{\mathsf T}=0$, and similarly $H_Z(\Hres_X)^{\mathsf T}=0$, a contradiction.
If all pairs $(i,j)\in [L/2]^2\setminus \Delta$ commute, then each term $F_uG_{r-u}+G_{r-u}F_u$ vanishes over $\Ftwo$, so $\Psi_r=0$ for all $r\notin\Delta$.
\end{proof}

\subsection{Upper bound on girth}
Finally, we recall a structural limitation that applies to the active matrices when commutativity holds on all required differences.
\begin{theorem}
Let $L$ be even and assume $3\le J\le L/2$ with the standard active choice (top $J$ block rows). If $F_u$ and $G_{r-u}$ commute for all $r\in\Delta$ and all $u\in[L/2]$, then the Tanner graphs of the active matrices $H_X$ and $H_Z$ each contain an 8-cycle.
\end{theorem}
\begin{proof}
Since $J\ge 3$, the active block rows include $0,1,2$. Choose $i,j\in[L/2]$ so that $i+j,i+j-1,i+j-2\in\Delta$ (e.g., $i=j=0$). Consider the block positions on the active rows
\[
[(0,i),(0,L/2+j),(1,L/2+j),(1,i+1),(2,i+1),(2,L/2+j+1),(1,L/2+j+1),(1,i)],
\]
where indices are modulo $L/2$. The corresponding cycle product, i.e., the product of permutation matrices encountered along this closed walk, is
\[
W=F_iG_j^{-1}G_{j-1}F_i^{-1}F_{i-1}G_{j-1}^{-1}G_jF_{i-1}^{-1}
\label{eq:8cycle-word-general}
\]
and the pairs $(i,j)$, $(i,j-1)$, $(i-1,j)$, and $(i-1,j-1)$ correspond to the differences $i+j$, $i+j-1$, and $i+j-2$ in $\Delta$. By assumption the relevant $F$ and $G$ commute, so $W=I$ and an 8-cycle exists in the active Tanner graphs.
\end{proof}

\section{Method of Matrix Construction}
\label{sec:sequential-apm}
Building on the commutation conditions and active orthogonality derived in the previous section, we now present a concrete sequential procedure to construct permutation blocks that satisfy those constraints while avoiding short cycles.
For a given $(L/2,P,J)$, we construct $\{F_i\}$ and $\{G_i\}$ that satisfy the following simultaneously:
\begin{itemize}
\item[(1)] $F_i$ and $G_j$ commute for $(i,j)\in \Gamma$.
\item[(2)] At least one $(i,j)\in [L/2]^2\setminus \Gamma$ is non-commuting.
\item[(3)] Avoid short cycles in the active matrices $H_X$ and $H_Z$.
\end{itemize}

The framework of constructing   LDPC blocks by CPM lifting  was systematized by Fossorier~\cite{Fossorier2004}. We adopt   APMs. The APM--LDPC framework extends CPMs, and its algebraic form makes commutativity control straightforward via congruence conditions~\cite{Gholami2015}. In the classical setting, Yoshida and Kasai reported that linear permutation polynomial (APM-based) codes achieve performance comparable to protograph LDPC codes~\cite{YoshidaKasai2019}, while no fixed girth upper bound has been reported for APM-based constructions, unlike one-step CPM lifting. Methods that combine APM--LDPC codes to extend length and girth have also been proposed~\cite{myung2006combining}. From a group-theoretic viewpoint, APMs on $\ZP$ form the affine group $\mathrm{AGL}(1,\ZP)=\ZP\rtimes\ZPx$, a semidirect product of translations by the unit group; commutativity is therefore governed by the interaction between the linear and translation parts.

Consider an affine permutation (AP) on $\ZP$,
\begin{equation*}
  f(x)=ax+b,\qquad a\in\ZPx,\ b\in\ZP.
\end{equation*}
We set $f_i(x)=a_i x+b_i$ and $g_j(x)=c_j x+d_j$, and denote the corresponding permutation matrices by $F_i:=P(f_i)$ and $G_j:=P(g_j)$.

Under this representation, commutativity of APs can be checked by a quadratic congruence:
\[
  f_i g_j = g_j f_i
  \iff
  d_j(a_i-1) - b_i(c_j-1) \equiv 0 \pmod P,
  \tag{C(i,j)}
\]
See~\cite{Gholami2015} for a derivation. The condition involves products of $a_i,c_j$ and $b_i,d_j$~\cite{Gholami2015}. The commutation table can thus be expressed as a system of congruences; in particular, if $a_i,c_j$ are chosen first, the commutation constraints reduce to linear congruences in $b_i,d_j$, enabling a consistent search~\cite{Gholami2015,myung2006combining}.

We sequentially select candidates that satisfy both the commutation table and short-cycle conditions, using backtracking when a candidate fails. The number of trials is adjusted dynamically based on recent success rates. We interpret each candidate generator as an arm in a multi-armed bandit and allocate trials accordingly~\cite{Auer2002}.
 {The search should be regarded as an offline design step rather than as part of encoding or decoding.
For fixed $(J,L)$, the expensive tests are finite congruence checks for the required commutation table and fixed-point tests for candidate short cycles.
They are much cheaper than enumerating the full Tanner graph, but the number of attempted APM tuples can still grow quickly when $P$ is small or when many non-commuting pairs are prescribed.
In practice we therefore use the bandit rule only as a heuristic for allocating trials among candidate generators; it does not provide a polynomial-time guarantee for finding an instance.
Once an instance is found, however, the resulting code is described by the small table of affine maps and has the usual sparse-matrix encoding and decoding complexity.}

Short-cycle detection reduces to checking fixed points of the composite map $\Sigma$ along a block-cycle pattern~\cite{Fossorier2004,Gholami2015}. The AP composition is again AP, $\Sigma(x)=Ax+B$, so
\begin{equation*}
  \Sigma(x)=x\ \Longleftrightarrow\ \gcd(A-1,P)\mid B
\end{equation*}
provides a test without enumerating the Tanner graph~\cite{Gholami2015}.

\section{Method of Decoding}
Our experiments use BP decoding~\cite{MacKay2004,PoulinChung2008} as the baseline.
 {For the depolarizing channel we use a joint BP update that keeps the $X$ and $Z$ components correlated through the qubit error probabilities, rather than running two completely independent binary BP decoders; this joint binary formulation is equivalent to four-state BP after a relabeling of the local Pauli states~\cite{Kasai2026JointFourStateBP}.
This distinction matters near the waterfall region because a separate binary treatment discards the $Y$-error correlation information.
The difference is not only a waterfall-region issue: in CSS-code studies based on non-binary LDPC codes, analogous decoder-design choices have been observed to make differences large enough to affect the error-floor behavior~\cite{KasaiHagiwaraImaiSakaniwa2012,KomotoKasai2025}.
The post-processing described below is only invoked after BP stalls or oscillates with a small residual syndrome, so the reported FER is not a result of an exhaustive global search after every frame.}
 We first recall a basic principle: if we knew the support where the BP estimate differs from the true error, decoding would reduce to a linear system on the corresponding columns, and a unique solution would imply successful recovery. We write $H[A,B]$ for the submatrix with row set $A$ and column set $B$, and for a vector $\boldsymbol{v}$ we write $\boldsymbol{v}[A]$ for its restriction to $A$. For example, if the $X$-side mismatch support is $E$, then with $A$ the full check set, $H_Z[A,E]\,\boldsymbol{e}[E]=\boldsymbol{r}_X[A]$; if $H_Z[A,E]$ has full column rank, the solution is unique and yields the correct correction. The same applies to the $Z$ side with $H_X$ and $\boldsymbol{r}_Z$.

Let the BP estimate at each iteration be $\boldsymbol{\hat{x}},\boldsymbol{\hat{z}}$. Define the residual syndromes
\begin{equation*}
\boldsymbol{r}_X:=\boldsymbol{s}_X\oplus H_Z\boldsymbol{\hat{x}},\qquad
\boldsymbol{r}_Z:=\boldsymbol{s}_Z\oplus H_X\boldsymbol{\hat{z}}.
\end{equation*}
Empirically, BP stalls or slowly oscillates roughly once per $10^5$ trials (frames), and we trigger post-processing in such cases, focusing on instances with a small number of unsatisfied checks; in our experiments we activate it only when the number of unsatisfied checks is at most 20.

We invoke three post-processing methods:  an elementary-trapping-set (ETS) library, flip-history decoding (FHD), and an ordered-statistics-decoding (OSD) method~\cite{FossorierLin1995}. Below we briefly describe the procedure and the criteria. In post-processing, the residual syndrome $\boldsymbol{r}$ is $\boldsymbol{r}_X$ (with $H=H_Z$) or $\boldsymbol{r}_Z$ (with $H=H_X$), depending on the side.

 {These three stages have different roles.
The ETS library targets the dominant small trapping-set patterns found from the Tanner graph.
FHD uses the empirical flip history of the stalled BP trajectory and is therefore more adaptive but still local.
The OSD-based stage is the most general of the three, because it uses BP reliabilities to choose a candidate support; it is also the most expensive and is used only when the residual syndrome remains small.
This staged design is intended to separate the code-family effect from the post-processing burden: most frames are decided by BP, while post-processing resolves rare small-residual stalls.}

{From an ablation viewpoint, the measured finite-length performance should not be attributed to a single ingredient.
The girth condition removes 4-cycles and 6-cycles, reducing the shortest local correlations and the smallest elementary trapping structures~\cite{Tanner1981,Richardson2003ErrorFloors}.
The regular degree distribution fixes the local message-passing environment; in the present instance this is a $(3,12)$-regular choice, which keeps variable degree small but uses relatively heavy checks~\cite{RichardsonShokrollahiUrbanke2001,RichardsonUrbanke2008}.
The randomized APM search and the subsequent trapping-set filtering select among many algebraically admissible candidates, and this selection affects the remaining harmful elementary trapping sets~\cite{Gholami2015,Auer2002,Richardson2003ErrorFloors}.
Finally, the lift size $P$, and therefore the block length, gives enough room to satisfy active orthogonality, girth 8, and a favorable trapping-set spectrum simultaneously.
Thus the reported performance reflects the interaction of graph expansion at short scales, the fixed local degree profile, the randomized APM search, and the available block length, rather than a single isolated design choice.}

 Because the true mismatch support $E$ is not directly observable by the decoder, our approach is to find a candidate support $\hat{E}\supseteq E$ and solve the local system on its neighborhood; when the system defined by $H[N(\hat{E}),\hat{E}]$ and $\boldsymbol{r}[N(\hat{E})]$ has a unique solution $\boldsymbol{\delta}[\hat{E}]$, we adopt it as the mismatch correction. 
The correction is applied only when the solution is unique and its weight is below a threshold.
Below we summarize the heuristics used to estimate $\hat{E}$.

\subsection{FHD}
FHD  uses the flip history to form a candidate support $\hat{E}=F$, where $F$ is the set of variables that flipped during BP, and then applies the local linear solve on $N(F)$. 
This local correction is a heuristic to resolve stalls caused by trapping structures~\cite{Richardson2003ErrorFloors}. Related local-flip decoders include weighted bit-flipping~\cite{ZhangFossorier2004}. 

\subsection{OSD}
OSD  uses the least reliable variables from BP to form a candidate support $\hat{E}=K$. 
In the implementation, we perform a binary search on $|K|$ to find the minimum $K$ that is solvable, then take the largest $K$ within the range that preserves uniqueness.

\subsection{ETS-based post-processing}
An $(a,b)$   ETS  is a variable-node set $V$ with $|V|=a$ whose induced subgraph has exactly $b$ unsatisfied (odd-degree) check nodes; it is elementary if every check node in the induced subgraph has degree 1 or 2.
The post-processing targets   ETSs  with $b=2$ unsatisfied checks, which are among the most harmful structures that cause BP to stall. 
We precompute a library of harmful trapping sets from the Tanner graph and store each entry as a variable set $V$ (used as $\hat{E}$) and its odd check pair $(c_0,c_1)$. 
During decoding, we run the post-processing only when the residual syndrome has exactly two unsatisfied checks.

\section{Results}
\label{sec:example}

We instantiate the general theory of Section~\ref{sec:general-theory} and the sequential construction of Section~\ref{sec:sequential-apm} for the smallest case we were able to construct, $J=3,L=12,P=768$. We first give explicit $\Delta,\Gamma$ and $\Psi_r$, clarify the active orthogonality condition for the active matrices $H_X$ and $H_Z$, then design the commutation table and APM parameters, and finally evaluate distance bounds and FER.

\subsection{Structure and active orthogonality}
For $J=3, L/2=6$, each block row is a cyclic shift of the previous row, so $\Hful_X,\Hful_Z$ have a $6\times 12$ block-circulant structure:
\begin{align}
 \begin{aligned} \Hful_X= & \left(\begin{array}{llllll|llllll}
 F_0 & F_1 & F_2 & F_3 & F_4 & F_5 & G_0 & G_1 & G_2 & G_3 & G_4 & G_5 \\
 F_5 & F_0 & F_1 & F_2 & F_3 & F_4 & G_5 & G_0 & G_1 & G_2 & G_3 & G_4 \\
 F_4 & F_5 & F_0 & F_1 & F_2 & F_3 & G_4 & G_5 & G_0 & G_1 & G_2 & G_3 \\\hline 
 F_3 & F_4 & F_5 & F_0 & F_1 & F_2 & G_3 & G_4 & G_5 & G_0 & G_1 & G_2 \\
 F_2 & F_3 & F_4 & F_5 & F_0 & F_1 & G_2 & G_3 & G_4 & G_5 & G_0 & G_1 \\
 F_1 & F_2 & F_3 & F_4 & F_5 & F_0 & G_1 & G_2 & G_3 & G_4 & G_5 & G_0\end{array}\right) .\\
\Hful_Z= & \left(\begin{array}{llllll|llllll}
 G_0^{\prime} & G_5^{\prime} & G_4^{\prime} & G_3^{\prime} & G_2^{\prime} & G_1^{\prime} & F_0^{\prime} & F_5^{\prime} & F_4^{\prime} & F_3^{\prime} & F_2^{\prime} & F_1^{\prime} \\
 G_1^{\prime} & G_0^{\prime} & G_5^{\prime} & G_4^{\prime} & G_3^{\prime} & G_2^{\prime} & F_1^{\prime} & F_0^{\prime} & F_5^{\prime} & F_4^{\prime} & F_3^{\prime} & F_2^{\prime} \\
 G_2^{\prime} & G_1^{\prime} & G_0^{\prime} & G_5^{\prime} & G_4^{\prime} & G_3^{\prime} & F_2^{\prime} & F_1^{\prime} & F_0^{\prime} & F_5^{\prime} & F_4^{\prime} & F_3^{\prime} \\\hline
 G_3^{\prime} & G_2^{\prime} & G_1^{\prime} & G_0^{\prime} & G_5^{\prime} & G_4^{\prime} & F_3^{\prime} & F_2^{\prime} & F_1^{\prime} & F_0^{\prime} & F_5^{\prime} & F_4^{\prime} \\
 G_4^{\prime} & G_3^{\prime} & G_2^{\prime} & G_1^{\prime} & G_0^{\prime} & G_5^{\prime} & F_4^{\prime} & F_3^{\prime} & F_2^{\prime} & F_1^{\prime} & F_0^{\prime} & F_5^{\prime} \\
 G_5^{\prime} & G_4^{\prime} & G_3^{\prime} & G_2^{\prime} & G_1^{\prime} & G_0^{\prime} & F_5^{\prime} & F_4^{\prime} & F_3^{\prime} & F_2^{\prime} & F_1^{\prime} & F_0^{\prime}
\end{array}\right) .\end{aligned}
\end{align}
Here, for simplicity, $A'$ denotes the transpose of $A$.
For example, the $(0,1)$ block of $\Hful_X(\Hful_Z)^{\mathsf T}$ is
\begin{equation*}
\begin{aligned}
(\Hful_X(\Hful_Z)^{\mathsf T})_{0,1}
&= F_0G_1+F_1G_0+F_2G_5+F_3G_4+F_4G_3+F_5G_2 \\
&\quad +G_0F_1+G_1F_0+G_2F_5+G_3F_4+G_4F_3+G_5F_2 \\
&= \Psi_1.
\end{aligned}
\end{equation*}
With $\Delta=\{0,1,2,4,5\}$,
\[
\Gamma=\{(i,j)\in[L/2]^2\mid i+j\not\equiv 3\ (\bmod\, 6)\}
=[L/2]^2\setminus\{(0,3),(1,2),(2,1),(3,0),(4,5),(5,4)\}.
\]
For the active matrices $H_X$ and $H_Z$ with $J=3$,
\begin{equation*}
H_XH_Z^{\mathsf T}
=
\bigl(\Psi_{(k-i)\bmod\, 6}\bigr)_{0\le i,k\le 2}
=
\begin{pmatrix}
\Psi_0 & \Psi_1 & \Psi_2\\
\Psi_5 & \Psi_0 & \Psi_1\\
\Psi_4 & \Psi_5 & \Psi_0
\end{pmatrix}.
\end{equation*}
Here we set $r:=(k-i)\bmod\,6$. For the active part we have $i,k\in\{0,1,2\}$, hence
$r\in\Delta=\{0,1,2,4,5\}$.
Thus active orthogonality is equivalent to
\begin{equation*}
\Psi_r=0\qquad(r\in\{0,1,2,4,5\}=\Delta),
\end{equation*}
so the $r=3$ interaction can be kept free. The parent-matrix product $\Hful_X(\Hful_Z)^{\mathsf T}$ is
\begin{equation*}
\Hful_X(\Hful_Z)^{\mathsf T}
=
\bigl(\Psi_{(k-i)\bmod\, 6}\bigr)_{0\le i,k\le 5}
=
\left(
\begin{array}{lll|lll}
\Psi_0 & \Psi_1 & \Psi_2 & \Psi_3 & \Psi_4 & \Psi_5\\
\Psi_5 & \Psi_0 & \Psi_1 & \Psi_2 & \Psi_3 & \Psi_4\\
\Psi_4 & \Psi_5 & \Psi_0 & \Psi_1 & \Psi_2 & \Psi_3\\ \hline
\Psi_3 & \Psi_4 & \Psi_5 & \Psi_0 & \Psi_1 & \Psi_2\\
\Psi_2 & \Psi_3 & \Psi_4 & \Psi_5 & \Psi_0 & \Psi_1\\
\Psi_1 & \Psi_2 & \Psi_3 & \Psi_4 & \Psi_5 & \Psi_0
\end{array}
\right).
\end{equation*}

\subsection{Design goal and commutation table}
With $L/2=6$ and $J=3$, we have $\Delta=\{0,1,2,4,5\}$, so that $3\notin\Delta$.
Our design guideline is to enlarge $d^{(\mathrm{lat})}$ as a proxy for increasing $d_{\min}$. In this setting, we target
\[
\Psi_r=0\quad (r\in\Delta), \qquad \Psi_3\neq 0.
\]
This condition guarantees active orthogonality while deliberately relaxing parent orthogonality,
thereby creating room to prevent the latent part from becoming orthogonal to the active part.

As an illustrative commutation table, suppose that all pairs $(F_i,G_j)$ commute except
$(i,j)=(0,3)$ and $(1,2)$.
The table below indicates the commuting pairs in $\Gamma$ by~$\Creq$:
\[
\begin{array}{c|cccccc}
 & G_0 & G_1 & G_2 & G_3 & G_4 & G_5\\ \hline
F_0 & \Creq & \Creq & \Creq & 0 & \Creq & \Creq\\
F_1 & \Creq & \Creq & 0 & \Creq & \Creq & \Creq\\
F_2 & \Creq & \Creq & \Creq & \Creq & \Creq & \Creq\\
F_3 & \Creq & \Creq & \Creq & \Creq & \Creq & \Creq\\
F_4 & \Creq & \Creq & \Creq & \Creq & \Creq & \Creq\\
F_5 & \Creq & \Creq & \Creq & \Creq & \Creq & \Creq
\end{array}
\]
Under this commutation pattern, we obtain
\[
\Psi_r=0\quad (r\neq 3), \qquad
\Psi_3=F_0G_3+G_3F_0+F_1G_2+G_2F_1.
\]

At first glance, one might expect that it is sufficient to keep $\Psi_3$ nonzero by making only
$F_0$ and $G_3$ non-commuting.
Indeed, we initially considered this minimal design and constructed codes satisfying $\Psi_3\neq 0$.
However, this choice resulted in many $(a,b=2)$   ETSs  formed by connecting multiple
length-8 cycles, which are among the most harmful structures for BP decoding~\cite{Richardson2003ErrorFloors}.
Moreover, two trapping sets sharing these same check-node pairs gave rise to low-weight logical operators.
Representative examples are shown in Fig.~\ref{fig:ets-8cycles}.
In contrast, by allowing both $(F_0,G_3)$ and $(F_1,G_2)$ to be non-commuting, we introduce additional
degrees of freedom in $\Psi_3$, which substantially reduces the number of harmful trapping sets while still
satisfying the active orthogonality constraint.
In example codes constructed this way, the only small trapping sets of weight $a$ are those listed in Section~\ref{sec:ets-library},
and none of them produced low-weight logical operators.

\paragraph{Remark.}
Throughout Section~\ref{sec:general-theory} we considered the standard active choice consisting of the top $J$ block
rows.
More generally, let $S\subseteq\mathbb{Z}_{L/2}$ be an arbitrary active index set with $|S|=J$,
and define
\[
\Delta_S := \{(k-i)\bmod (L/2) \mid i,k\in S\}.
\]
The proof of Theorem~\ref{thm:active-orthogonality} applies verbatim with $\Delta$ replaced by $\Delta_S$.

This generalization allows a substantial reduction of $|\Delta|$.
In particular, for $J=3$ and $L=12$, choosing $S=\{0,2,4\}$ gives $|\Delta_S|=3$,
whereas the standard choice yields $|\Delta|=5$.
Such a reduction increases the number of indices $r\notin\Delta_S$ for which $\Psi_r$
can be kept nonzero, and may be exploited to further control trapping-set structures.
Under this relaxed constraint, constructing codes may allow us to achieve the girth and minimum distance studied here at shorter code lengths.
 {More generally, the rate is controlled by the active fraction $2J/L$ and by rank deficiencies in the two active matrices.
The present paper focuses on the rate-near-one-half case because it gives a clean first example in which girth 8, column weight 3, and latent non-orthogonality can all be verified explicitly.
Other rates can be targeted by changing $J/L$ and the active set $S$, but then the commutation table, available non-commuting residues, and trapping-set profile must be redesigned together.
At smaller lengths the same algebraic conditions still apply, but the APM search becomes more constrained and we have not yet found a smaller instance with the same combination of girth, latent distance, and decoding performance.}

\subsection{Affine permutation construction and choice of \texorpdfstring{$P$}{P}}
We construct a girth-8 $(3,12)$-regular code; the parameters $f_i,g_i$ are listed in Table~\ref{tab:apm-fg-params} (indices start at 0). Each $f_i,g_i$ is an affine permutation $x\mapsto ax+b$ on $\ZP$ with $\gcd(a,P)=1$. The resulting Tanner graph has girth 8 (no 4- or 6-cycles). 
\begin{table}[t]
\centering
\caption{Constructed APM parameters ($P=768,J=3,L=12$, girth=8)}
\label{tab:apm-fg-params}
\renewcommand{\arraystretch}{1.1}
\small
\begin{tabular}{c|l|l}
\hline
$i$ & $f_i(x)$ & $g_i(x)$ \\
\hline
0 & $763x+435$ & $289x+496$ \\
1 & $679x+69$ & $257x+640$ \\
2 & $397x+330$ & $625x+200$ \\
3 & $61x+18$ & $41x+524$ \\
4 & $697x+612$ & $193x+672$ \\
5 & $373x+246$ & $449x+672$ \\
\hline
\end{tabular}
\end{table}
If $P$ is a prime power, there is no set of affine permutations $\sigma_0,\sigma_1,\tau_0,\tau_1$ on $[P]$ that simultaneously satisfy: $\sigma_0$ commutes with $\tau_0$ but not with $\tau_1$, and $\sigma_1$ does not commute with $\tau_0$ but commutes with $\tau_1$. Using the standard semidirect-product representation of $\mathrm{AGL}_1(\mathbb{Z}/p^k\mathbb{Z})$, $(a,b):x\mapsto ax+b$, and the quadratic commutation condition, one can show this by analyzing the resulting zero-product conditions with $p$-adic valuation, which contradict non-commutativity. Therefore we choose $P=768=3\times 2^8$.

\subsection{Short Cycles}
We begin by presenting a concrete example illustrating Theorem~3.4.
For the $J=3,L=12$, the block positions
\[
 [(0,0),(0,6),(1,6),(1,1),(2,1),(2,7),(1,7),(1,0)]
\label{eq:8cycle-example}
\]
form an 8-cycle when traversed in order. The corresponding cycle word, in the form of Eq.~\eqref{eq:8cycle-word-general}, is
\begin{equation*}
W=F_0G_0^{-1}G_5F_0^{-1}F_5G_5^{-1}G_0F_5^{-1}.
\end{equation*}
Counting such active block 8-cycles is also straightforward. For the active rows $0,1,2$, the unavoidable pattern uses adjacent left columns $i,i+1$ and adjacent right columns $j,j+1$ (indices modulo $L/2$), and it occurs when $i+j,i+j-1,i+j-2\in\Delta$. Thus the number of block 8-cycles in the active matrices is $(L/2)|\Delta\cap(\Delta+1)\cap(\Delta+2)|$. For $J=3,L/2=6$, $\Delta=\{0,1,2,4,5\}$ and $\Delta\cap(\Delta+1)\cap(\Delta+2)=\{0,1,2\}$, so there are 18 block 8-cycles, yielding $18P$ lifted 8-cycles.

Enumerating cycles in the active Tanner graphs of this instance gives $60{,}512$ 8-cycles on the $X$ side, and $54{,}656$ 8-cycles on the $Z$ side. By Theorem~3.4, 8-cycles are unavoidable in the active Tanner graphs; the block pattern above yields the lower bound $18P=13{,}824$. The total exceeds this bound because commutativity holds for many additional pairs and the specific APM parameters create extra cycles beyond block lifts, so the total need not be a multiple of $P$.

\subsection{Trapping set library construction}
\label{sec:ets-library}
In our initial experiments, trapping sets of the form shown in Fig.~\ref{fig:ets-8cycles} were the dominant cause of BP stalls. Since the constructed code has girth 8, we first enumerate 8-cycles and then build a trapping-set library by connecting several 8-cycles to generate trapping sets isomorphic to these patterns. The library is not a collection of only the trapping sets encountered in decoding; instead, we enumerate all trapping sets isomorphic to the dominant patterns and add them exhaustively.

Figure~\ref{fig:ets-8cycles} shows three elementary trapping-set patterns: $(6,2)$ from two length-8 cycles, $(12,2)$ from five length-8 cycles, and $(8,2)$ formed by attaching a length-4 check--variable path to the two odd checks of (6,2). For $P=768,J=3,L=12$, the counts were: $(6,2)$ X-side: 48, Z-side: 16 (total 64); $(12,2)$ X-side: 23, Z-side: 0 (total 23); and $(8,2)$ path4 X-side: 48, Z-side: 0 (total 48).

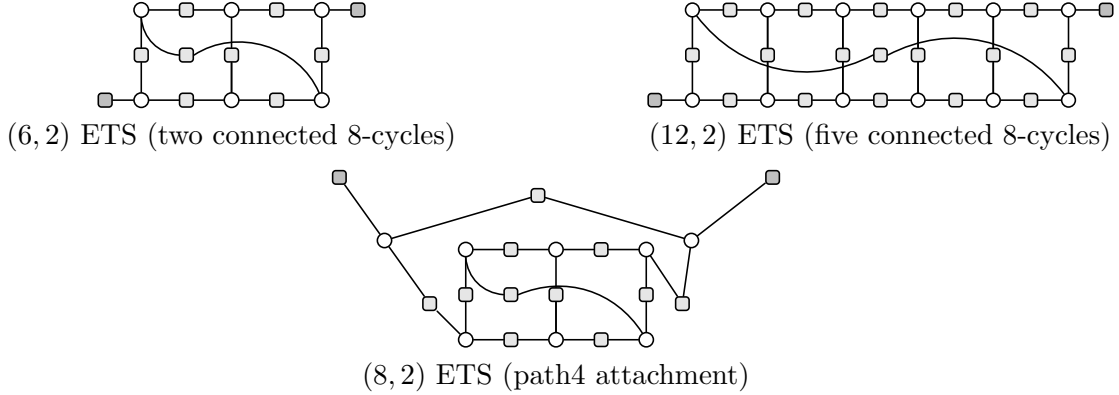
\begin{figure}[t]
\centering
\begin{minipage}{0.48\linewidth}
\centering
\begin{tikzpicture}[
  x=0.6cm,y=0.6cm,
  var/.style={circle,draw=black,fill=white,line width=0.6pt,minimum size=5.5pt,inner sep=0pt},
  chk/.style={rectangle,draw=black,fill=gray!20,line width=0.6pt,minimum size=5.2pt,inner sep=0pt,rounded corners=1.2pt},
  odd/.style={rectangle,draw=black,fill=gray!50,line width=0.6pt,minimum size=5.2pt,inner sep=0pt,rounded corners=1.2pt},
  edge/.style={line width=0.6pt,draw=black,line cap=round,line join=round}
]
\node[var] (v1) at (0,0) {};
\node[var] (v2) at (0,2) {};
\node[var] (v3) at (2,2) {};
\node[var] (v4) at (2,0) {};
\node[var] (v5) at (4,2) {};
\node[var] (v6) at (4,0) {};

\node[chk] (c1) at (1,2) {};
\node[chk] (c2) at (2,1) {};
\node[chk] (c3) at (1,0) {};
\node[chk] (c4) at (0,1) {};
\node[chk] (c5) at (3,2) {};
\node[chk] (c6) at (4,1) {};
\node[chk] (c7) at (3,0) {};
\node[chk] (c8) at (1.0,1.0) {};

\node[odd] (o1) at (-0.8,0) {};
\node[odd] (o2) at (4.8,2) {};

\draw[edge] (v2)--(c1)--(v3)--(c2)--(v4)--(c3)--(v1)--(c4)--(v2);
\draw[edge] (v3)--(c5)--(v5)--(c6)--(v6)--(c7)--(v4)--(c2);
\draw[edge] (v1)--(o1);
\draw[edge] (v5)--(o2);
\draw[edge] (v2) to[bend right=40] (c8.west);
\draw[edge] (c8.east) to[bend left=50] (v6);
\end{tikzpicture}
\par
 $(6,2)$ ETS (two connected 8-cycles)
\end{minipage}
\hfill
\begin{minipage}{0.48\linewidth}
\centering
\begin{tikzpicture}[
  x=0.5cm,y=0.6cm,
  var/.style={circle,draw=black,fill=white,line width=0.6pt,minimum size=5.2pt,inner sep=0pt},
  chk/.style={rectangle,draw=black,fill=gray!20,line width=0.6pt,minimum size=5.0pt,inner sep=0pt,rounded corners=1.2pt},
  odd/.style={rectangle,draw=black,fill=gray!50,line width=0.6pt,minimum size=5.0pt,inner sep=0pt,rounded corners=1.2pt},
  edge/.style={line width=0.6pt,draw=black,line cap=round,line join=round}
]
\node[var] (v1) at (0,2) {};
\node[var] (v2) at (0,0) {};
\node[var] (v3) at (2,2) {};
\node[var] (v4) at (2,0) {};
\node[var] (v5) at (4,2) {};
\node[var] (v6) at (4,0) {};
\node[var] (v7) at (6,2) {};
\node[var] (v8) at (6,0) {};
\node[var] (v9) at (8,2) {};
\node[var] (v10) at (8,0) {};
\node[var] (v11) at (10,2) {};
\node[var] (v12) at (10,0) {};

\node[chk] (u0) at (0,1) {};
\node[chk] (u1) at (2,1) {};
\node[chk] (u2) at (4,1) {};
\node[chk] (u3) at (6,1) {};
\node[chk] (u4) at (8,1) {};
\node[chk] (u5) at (10,1) {};

\node[chk] (t1) at (1,2) {};
\node[chk] (t2) at (3,2) {};
\node[chk] (t3) at (5,2) {};
\node[chk] (t4) at (7,2) {};
\node[chk] (t5) at (9,2) {};
\node[chk] (b1) at (1,0) {};
\node[chk] (b2) at (3,0) {};
\node[chk] (b3) at (5,0) {};
\node[chk] (b4) at (7,0) {};
\node[chk] (b5) at (9,0) {};
\node[chk] (u6) at (5,1.0) {};

\node[odd] (o3) at (-1,0) {};
\node[odd] (o4) at (11,2) {};

\draw[edge] (v1)--(t1)--(v3)--(u1)--(v4)--(b1)--(v2)--(u0)--(v1);
\draw[edge] (v3)--(t2)--(v5)--(u2)--(v6)--(b2)--(v4)--(u1)--(v3);
\draw[edge] (v5)--(t3)--(v7)--(u3)--(v8)--(b3)--(v6)--(u2)--(v5);
\draw[edge] (v7)--(t4)--(v9)--(u4)--(v10)--(b4)--(v8)--(u3)--(v7);
\draw[edge] (v9)--(t5)--(v11)--(u5)--(v12)--(b5)--(v10)--(u4)--(v9);
\draw[edge] (v2)--(o3);
\draw[edge] (v11)--(o4);
\draw[edge] (v1) to[bend right=40] (u6.west);
\draw[edge] (u6.east) to[bend left=40] (v12);
\end{tikzpicture}
\par
 $(12,2)$ ETS (five connected 8-cycles)
\end{minipage}
\par\medskip
\begin{minipage}{0.70\linewidth}
\centering
\begin{tikzpicture}[
  x=0.6cm,y=0.6cm,
  var/.style={circle,draw=black,fill=white,line width=0.6pt,minimum size=5.4pt,inner sep=0pt},
  chk/.style={rectangle,draw=black,fill=gray!20,line width=0.6pt,minimum size=5.1pt,inner sep=0pt,rounded corners=1.2pt},
  odd/.style={rectangle,draw=black,fill=gray!50,line width=0.6pt,minimum size=5.1pt,inner sep=0pt,rounded corners=1.2pt},
  edge/.style={line width=0.6pt,draw=black,line cap=round,line join=round}
]
\node[var] (v1) at (0,0) {};
\node[var] (v2) at (0,2) {};
\node[var] (v3) at (2,2) {};
\node[var] (v4) at (2,0) {};
\node[var] (v5) at (4,2) {};
\node[var] (v6) at (4,0) {};

\node[chk] (c1) at (1,2) {};
\node[chk] (c2) at (2,1) {};
\node[chk] (c3) at (1,0) {};
\node[chk] (c4) at (0,1) {};
\node[chk] (c5) at (3,2) {};
\node[chk] (c6) at (4,1) {};
\node[chk] (c7) at (3,0) {};
\node[chk] (c8) at (1.0,1.0) {};

\node[chk] (c599) at (-0.8,0.8) {};
\node[chk] (c542) at (4.8,0.8) {};

\node[var] (v8) at (-1.8,2.2) {};
\node[chk] (c893) at (1.6,3.2) {};
\node[var] (v7) at (5.0,2.2) {};

\node[odd] (o1715) at (-2.8,3.6) {};
\node[odd] (o1991) at (6.8,3.6) {};

\draw[edge] (v2)--(c1)--(v3)--(c2)--(v4)--(c3)--(v1)--(c4)--(v2);
\draw[edge] (v3)--(c5)--(v5)--(c6)--(v6)--(c7)--(v4)--(c2);
\draw[edge] (v1)--(c599);
\draw[edge] (v5)--(c542);
\draw[edge] (v2) to[bend right=40] (c8.west);
\draw[edge] (c8.east) to[bend left=40] (v6);

\draw[edge] (c542)--(v7)--(c893)--(v8)--(c599);
\draw[edge] (v8)--(o1715);
\draw[edge] (v7)--(o1991);
\end{tikzpicture}
\par
 $(8,2)$ ETS (path4 attachment)
\end{minipage}
\caption{Representative ETSs used in the library. }
\label{fig:ets-8cycles}
\end{figure}

\subsection{Minimum-distance evaluation}
Recall the latent-based distances $d_X^{(\mathrm{lat})}$ and $d_Z^{(\mathrm{lat})}$ defined in Section~\ref{sec:deleted-checks-logicals}. We now construct explicit low-weight vectors in the latent row spaces that satisfy the active parity checks but are not in the duals. Equivalently, any $\boldsymbol{x}\in\Row(\Hres_X)$ can be written as $\boldsymbol{x}=(\Hres_X)^{\mathsf T}\boldsymbol{u}$ for some $\boldsymbol{u}$, and likewise for $Z$.
Concretely, we look for coefficient vectors $\boldsymbol{u}\in\Ker\!\bigl(H_Z(\Hres_X)^{\mathsf T}\bigr)$ and $\boldsymbol{v}\in\Ker\!\bigl(H_X(\Hres_Z)^{\mathsf T}\bigr)$, then lift them to $\boldsymbol{x}=\boldsymbol{u}^{\mathsf T}\Hres_X$ and $\boldsymbol{z}=\boldsymbol{v}^{\mathsf T}\Hres_Z$. Any such $\boldsymbol{x},\boldsymbol{z}$ yields $d_X^{(\mathrm{lat})}\le \wt(\boldsymbol{x})$ and $d_Z^{(\mathrm{lat})}\le \wt(\boldsymbol{z})$, so below we exploit the block-diagonal form to obtain explicit weight-48 examples.

Let $\boldsymbol{u},\boldsymbol{v}\in\Ftwo^{3P}$ be coefficient vectors satisfying
\begin{align*}
\diag(\Psi_3^{\mathsf T},\Psi_3^{\mathsf T},\Psi_3^{\mathsf T})\boldsymbol{u}&=0,\\
\diag(\Psi_3,\Psi_3,\Psi_3)\boldsymbol{v}&=0.
\end{align*}
Define $\boldsymbol{x}:=\boldsymbol{u}^{\mathsf T}\Hres_X$ and $\boldsymbol{z}:=\boldsymbol{v}^{\mathsf T}\Hres_Z$ using the latent matrices $\Hres_X$ and $\Hres_Z$. Here $H_X$ and $H_Z$ are the active matrices, so
$\boldsymbol{u}\in\Ker\!\bigl(H_Z(\Hres_X)^{\mathsf T}\bigr)$ and
$\boldsymbol{v}\in\Ker\!\bigl(H_X(\Hres_Z)^{\mathsf T}\bigr)$.
With $L/2=6,J=3$ and $\Psi_r=0\ (r\neq 3)$,
\[
H_Z(\Hres_X)^{\mathsf T}=\diag(\Psi_3^{\mathsf T},\Psi_3^{\mathsf T},\Psi_3^{\mathsf T}),\qquad
H_X(\Hres_Z)^{\mathsf T}=\diag(\Psi_3,\Psi_3,\Psi_3).
\label{eq:block-diag-psi3}
\]
Thus $\boldsymbol{u}=(\boldsymbol{u}_3,\boldsymbol{u}_4,\boldsymbol{u}_5)$ and
$\boldsymbol{v}=(\boldsymbol{v}_3,\boldsymbol{v}_4,\boldsymbol{v}_5)$ with
$\boldsymbol{u}_r\in\Ker(\Psi_3^{\mathsf T})$ and $\boldsymbol{v}_r\in\Ker(\Psi_3)$ for $r=3,4,5$.
Commutativity makes the terms for $u=2,3,4,5$ vanish, so
\[
\Psi_3=F_0G_3+G_3F_0+F_1G_2+G_2F_1,
\]
and hence
\[
\Psi_3^{\mathsf T}=F_0^{\mathsf T}G_3^{\mathsf T}+G_3^{\mathsf T}F_0^{\mathsf T}
+F_1^{\mathsf T}G_2^{\mathsf T}+G_2^{\mathsf T}F_1^{\mathsf T}.
\]
In this instance, $\mathrm{rank}(\Psi_3)=576$, so $\dim \Ker(\Psi_3)=\dim \Ker(\Psi_3^{\mathsf T})=192$.
The four permutation terms in $\Psi_3$ have disjoint supports, so both row and column weights are $4$ and the kernels are generated by $192$ disjoint weight-$4$ codewords supported on the $4$-point blocks
\[
[t]:=\{t,\ t+192,\ t+384,\ t+576\}\subset\Z_{768},\qquad t\in\Z_{192}.
\]
The explicit weight-$4$ vectors constructed above lift to $\boldsymbol{x}=\boldsymbol{u}^{\mathsf T}\Hres_X$ and $\boldsymbol{z}=\boldsymbol{v}^{\mathsf T}\Hres_Z$ of weight $48$, and they are not in $C_X^{\perp}$ and $C_Z^{\perp}$ for this instance. This gives the upper bounds $d_X^{(\mathrm{lat})} \le 48$ and $d_Z^{(\mathrm{lat})} \le 48$.

Because each APM block preserves the $4$-point blocks, any latent vector $\boldsymbol{x}=\boldsymbol{u}^{\mathsf T}\Hres_X$
is block-constant within each column block. Compressing each column block to length $192$ gives $\bar{\boldsymbol{x}}$ with
$\wt(\boldsymbol{x})=4\,\wt(\bar{\boldsymbol{x}})$. In the compressed space, each generator has weight $12$ and overlaps across
row blocks are uniformly bounded, so any nonzero $\bar{\boldsymbol{x}}$ has $\wt(\bar{\boldsymbol{x}})\ge 12$.
Hence $\wt(\boldsymbol{x})\ge 48$, and the same argument applies to $Z$, so $d_X^{(\mathrm{lat})}\ge 48$ and
$d_Z^{(\mathrm{lat})}\ge 48$ are provable here. Together with the explicit weight-$48$ examples we conclude
$d_X^{(\mathrm{lat})}=d_Z^{(\mathrm{lat})}=48$.
This shows $d_{\min}\le 48$. 

To conclude $d_{\min}=48$, we must also bound logical operators outside the latent range; in other words,
  we need lower bounds on the minimum weights of non-latent representatives, e.g., $\boldsymbol{x}\in C_Z\setminus C_X^\perp$ with $\boldsymbol{x}\notin\Row(\Hres_X)$, and similarly on the $Z$ side.

At present we do not have certified lower bounds on these non-latent minima, so the rigorous statement here is only $d_{\min}\le 48$. In summary, we prove $d_X^{(\mathrm{lat})}=d_Z^{(\mathrm{lat})}=48$ for this instance, which yields the bound $d_{\min}\le 48$, while certified lower bounds for non-latent logical operators remain open.

 {The role of girth in this distance behavior is indirect.
Girth 8 removes 4- and 6-cycles from the active Tanner graphs and is important for BP performance and for suppressing the smallest trapping structures.
It does not by itself certify a distance lower bound for the CSS code.
In classical regular LDPC codes, known results relate Tanner girth and minimum distance, but these relations are loose as finite-length distance certificates~\cite{Tanner2001}.
We use the standard regular-LDPC heuristic that larger girth is needed to suppress small cycles and trapping structures when seeking large-distance regular codes~\cite{Tanner1981,Richardson2003ErrorFloors}, rather than treating girth as a proof of distance.
The verified weight-48 latent operators instead come from the interaction matrix $\Psi_3$ and from the way active orthogonality is localized while mixed active--latent orthogonality is deliberately avoided.
Thus the present evidence for large distance is a combination of large girth, the rank and support structure of $\Psi_3$, and the absence of observed low-weight non-latent logical operators in the performed searches; only the latent part is currently certified.}

 \subsection{Frame error rate}
Figure~\ref{fig:fer-apm} reports the FER of the constructed code under BP with post-processing. For reference, we also include the hashing bound for the depolarizing channel in the plot. At $p=4\%$, the FER reaches $10^{-8}$. For each point we collected at least 50 error events; 95\% confidence intervals are plotted but may be visually indistinguishable because they are narrow. For FER above $10^{-7}$, all failures left hundreds of remaining errors, and we did not observe cases where the decoder returned low-weight logical errors (including higher-weight ones), which is common in small-distance codes. Around $10^{-8}$, an error-floor tendency begins to appear, and the dominant failures are stalls trapped by trapping sets of size on the order of tens.
 {The remaining gap to the hashing bound is therefore not attributed to a known low-weight-distance mechanism in this instance.
Rather, our current interpretation is that it comes from a combination of finite-length effects, residual trapping sets allowed by unavoidable 8-cycles, and the suboptimality of practical BP-type decoding with local post-processing relative to maximum-likelihood decoding.
}

\begin{figure}[t]
\centering
\includegraphics[width=0.85\linewidth]{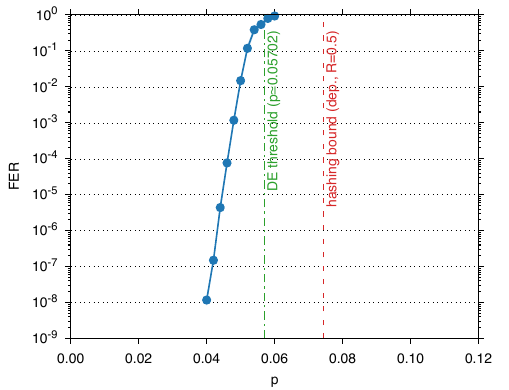}
\caption{FER curve of the constructed girth-8, (3,12)-regular \([[9216,4612, \leq 48]]\) code ($P=768,J=3,L=12$), with error bars (95\% confidence). For reference, we plot the depolarizing-channel hashing bound and the density-evolution benchmark threshold for a non-orthogonal random $(3,12)$-regular $H_X,H_Z$ pair ($p=0.05702$). Here the number of physical qubits is $n=PL$, the active matrix ranks are $\mathrm{rank}(H_X)=2302$ and $\mathrm{rank}(H_Z)=2302$, hence the number of logical qubits is $k=4612$.}
\label{fig:fer-apm}
\end{figure}

\subsection{Density evolution benchmark}
To benchmark BP performance, we analyze a non-orthogonal random $(3,12)$-regular pair of classical LDPC codes $H_X,H_Z$ drawn independently from the standard configuration-model ensemble. Although this pair is no longer a CSS code and has no direct physical meaning, it provides a useful classical benchmark for estimating noise from syndromes.
Under the cycle-free (infinite-length) approximation, we estimate the density-evolution  (DE)  fixed point using a population-dynamics (Monte Carlo) simulation: following the numerical viewpoint emphasized by MacKay~\cite{MacKay2003}, we track a large population of messages when analytic density tracking is intractable, iteratively updating it by randomly sampling tree-like neighborhoods with the regular degrees and applying the BP update rules for the depolarizing channel; see~\cite{RichardsonUrbanke2001,RichardsonUrbanke2008} for standard   DE  background. This yields a depolarizing-channel BP threshold of $p\simeq 0.05702$, which we plot in Fig.~\ref{fig:fer-apm} as a benchmark reference. The proposed code approaches this DE threshold, suggesting it retains sufficient LDPC randomness and that the decoder effectively exploits it.

\subsection{Instance dependence}
{For fixed local degrees and girth constraints, the waterfall region is expected to be largely governed by local message-passing behavior.
Although the concentration theorem for standard LDPC ensembles does not directly prove concentration for the present structured APM family, it supports the expectation that similar concentration around the ensemble-average behavior may occur in the waterfall region~\cite{RichardsonUrbanke2001}.
Thus we do not expect large instance-to-instance variation there.
By contrast, finite-instance effects become important in the minimum distance and in the error-floor region, where residual trapping sets and low-weight logical operators can dominate the observed FER.
This statement should be read on the logarithmic FER scale: an absolute FER difference of order \(10^{-8}\) is essentially invisible in the waterfall region, but it is significant when evaluating an error floor near FER \(10^{-8}\).
At first glance, one might expect that it is sufficient to keep $\Psi_3$ nonzero by making only $F_0$ and $G_3$ non-commuting.
Indeed, we initially considered this minimal design and constructed codes satisfying $\Psi_3\neq0$.
However, this choice resulted in many $(a,2)$ ETSs formed by connecting multiple length-8 cycles, which are among the most harmful structures for BP decoding.
Moreover, two trapping sets sharing these same check-node pairs gave rise to low-weight logical operators.
Allowing two non-commuting pairs in $\Psi_3$ gave additional degrees of freedom and led to the instance reported here.
Thus the construction should be viewed as a code family with nontrivial instance selection: the algebraic active-orthogonality conditions define the admissible search space, but simulation-based screening and distance searches are needed to reject candidates with small-distance mechanisms or poor error-floor behavior.}

\section*{Acknowledgements}
This work was supported by JSPS KAKENHI Grant Number JP26K07504.

\section*{Author contributions}
Kenta Kasai carried out all aspects of the research and manuscript preparation.
OpenAI Codex, a large language model, was used after acceptance to assist with reformatting the manuscript in the \texttt{quantumarticle} class, checking bibliographic metadata and DOI links, and drafting this author-contribution and AI-use disclosure. The author reviewed all AI-assisted output and takes full responsibility for the content of the manuscript.

\section*{Data availability}
The data that support the findings of this study are available from the author upon reasonable request.

\section*{Code availability}
\begin{sloppypar}
The construction algorithm is available at \url{https://github.com/kasaikenta/construct_apm_css_code}.
The decoding algorithm is available at \url{https://github.com/kasaikenta/joint_BP_plus_PP}.
\end{sloppypar}
 \bibliographystyle{quantum}
\bibliography{commute_quantum_final}
\end{document}